\documentclass[sigconf]{acmart}

\usepackage{tabularx,ragged2e}

\usepackage{graphicx}

\usepackage[linesnumbered,ruled,vlined,algo2e]{algorithm2e}
\makeatletter
\newcommand{\removelatexerror}{\let\@latex@error\@gobble}
\makeatother

\usepackage{amsmath}
\usepackage{amsthm}
\usepackage{amssymb}
\usepackage{float}

\usepackage{subcaption}
\usepackage{multirow}
\usepackage{hhline}

\usepackage{url}
\usepackage{breakurl}
\usepackage{xcolor}
\definecolor{light-gray}{gray}{0.90}

\usepackage{tikz}
\usetikzlibrary{arrows,positioning,shadows}
\usetikzlibrary{backgrounds}
\usepackage{pgf-umlsd-mod}

\pgfdeclarelayer{bkg}
\pgfsetlayers{background,bkg,main}

\usepackage{enumitem}
\usepackage{comment}

\newtheorem{lem}{Lemma}

\theoremstyle{remark}
\newtheorem*{model}{Attacker Model}

\usepackage[normalem]{ulem}

\newcommand{\vf}{\mbox{\it vf}}
\newcommand{\buf}{\mbox{\it buf}}

\newcommand{\KGen}{\mathsf{KGen}}
\newcommand{\INITSTORE}{\mathsf{INIT.STORE}}
\newcommand{\INITOC}{\mathsf{INIT.OC}}
\newcommand{\ACCESSRW}{\mathsf{ACCESSRW}}
\newcommand{\ACCESSOC}{\mathsf{ACCESSOC}}

\newcommand{\Exp}{\mathrm{Exp}}
\newcommand{\inb}{\mathrm{inb}}
\newcommand{\sel}{\mathrm{sel}}
\newcommand{\out}{\mathrm{out}}

\newcommand{\adv}{\mathcal{A}}
\newcommand{\cdv}{\mathcal{C}}
\newcommand{\prob}[1]{\mathsf{Pr\Big[ #1 \Big]}}

\newcommand{\schemeacr}{CAOS}
\newcommand{\schemefull}{Concurrent-Access Obfuscated Store}
\newcommand{\schemesymb}{\mathcal{O}}

\copyrightyear{2019}
\acmYear{2019}
\setcopyright{acmlicensed}
\acmConference[SACMAT '19]{The 24th ACM Symposium on Access Control Models and Technologies}{June 3--6, 2019}{Toronto, ON, Canada}
\acmBooktitle{The 24th ACM Symposium on Access Control Models and Technologies (SACMAT '19), June 3--6, 2019, Toronto, ON, Canada}
\acmPrice{15.00}
\acmDOI{10.1145/3322431.3325101}
\acmISBN{978-1-4503-6753-0/19/06}

\renewcommand\footnotetextcopyrightpermission[1]{}
\begin{document}
\title{\schemeacr{}: \schemefull{}}

\author{Mihai Ordean}
\affiliation{%
  \institution{University of Birmingham, UK}
}
\email{m.ordean@cs.bham.ac.uk}

\author{Mark Ryan}
\affiliation{%
  \institution{University of Birmingham, UK}
}
\email{m.d.ryan@cs.bham.ac.uk}

\author{David Galindo}
\affiliation{%
  \institution{University of Birmingham and\\Fetch.AI, Cambridge, UK}
}
\email{d.galindo@cs.bham.ac.uk}

\begin{abstract}
This paper proposes \schemefull{} (\schemeacr{}), a construction for remote data storage that provides access-pattern obfuscation in a honest-but-curious adversarial model, while allowing for low bandwidth overhead and client storage. Compared to other approaches, the main advantage of \schemeacr{} is that it supports concurrent access without a proxy, for multiple read-only clients and a single read-write client. Concurrent access is achieved by letting clients maintain independent maps that describe how the data is stored. Even though the maps might diverge from client to client, the protocol guarantees that clients will always have access to the data. Efficiency and concurrency are achieved at the expense of perfect obfuscation: in \schemeacr{} the extent to which access patterns are hidden is determined by the resources allocated to its built-in obfuscation mechanism. To assess this trade-off we provide both a security and a performance analysis of \schemeacr{}. We additionally provide a proof-of-concept implementation\footnote{Available: \url{https://github.com/meehien/caos}}.
\end{abstract}



\maketitle

\section{Introduction}


Cloud computing has become an attractive solution for data storage. Unfortunately, current cloud computing architectures do not provide sufficient and reliable security for private and sensitive data. Even when encryption is used, malicious servers and operators can learn user access patterns and derive information based on them (e.g., data accessed more often can be assumed to be more important) \cite{cloud_security_alliance}.

One cryptographic primitive specifically designed to hide access patterns is Oblivious RAM (ORAM). This primitive was introduced by Goldreich and Ostrovsky \cite{goldreich1987towards, goldreich1996software} for the purposes of preventing software reverse engineering by hiding a program's access patterns to memory. The issue has since become important in the context of cloud computing, where clients and data-store servers often reside in different trust domains and trust between them cannot always be established. Modern ORAM schemes \cite{stefanov2013path, ren2015constants, devadas2016onion} are seen as viable options of addressing this problem. However, in real-life scenarios, even the best ORAMs can prove to be impractical \cite{senny2014blog}, mainly because of the high bandwidth requirements and/or client storage constraints. Another major limitation of modern ORAM constructions is that they are mainly restricted to having a single-client that connects to the  data-store server. This is because data in the store is accessed through a client maintained local structure (i.e. a map). Migrating from this model has proven difficult. Even small deviations \cite{williams2012privatefs, stefanov2013oblivistore}, such as allowing multiple clients to access the store through a proxy that acts as the single-client have been shown to have vulnerabilities \cite{sahin2016taostore}.

As ORAMs have been difficult to use in real-life, other specialised, and more efficient security primitives have been developed in the context of privacy preserving access to cloud-stored data e.g. searchable encryption (SE) schemes. SE uses either symmetric keys \cite{curtmola_searchable_2006,david_cash_highly-scalable_2013,seny_kamara_parallel_2013} or public keys \cite{shi2007multi} and allows clients to securely search cloud stored databases through precomputed ciphertexts called \textit{trapdoors}. SE schemes have low computational requirements from clients and are bandwidth efficient. However, prior work has shown that searchable encryption schemes leak significant amounts of information about their encrypted indexes when using attacks which combine access-pattern analysis, background information about data stored, and language-based word frequency knowledge \cite{islam2012access}. Incorporating changes and updates to the searched database is also a difficult process. Often schemes require the whole index to be regenerated for any the new information added \cite{curtmola_searchable_2006,david_cash_highly-scalable_2013}. Finally, SE schemes are restricted to search operations, actual data retrieval needs to happen through a private information retrieval protocol \cite{chor1995private} or an ORAM.

As such, an ideal system would have the general applicability and access-pattern privacy of ORAM (cloud storage which hides access patterns), and bandwidth efficiency and concurrent access capabilities similar to those of SE schemes. In this paper we take steps towards this direction by proposing a new design for a general-purposed secure storage with concurrency and bandwidth efficiency. However, the privacy guarantees we provide are not absolute. Instead, our protocol requires that users provision resources for access-pattern obfuscation, and the security guarantees depend on how much of these resources are available.

\subsection{Contributions}
This paper proposes \schemefull{} (\schemeacr{}), a  storage 
access protocol that can hide data access frequency and access patterns, while allowing for concurrent data access. Our main focus when designing \schemeacr{} is to obtain a bandwidth-efficient protocol that supports concurrency by design and that is able to provide a customizable amount of data and access-pattern privacy. Our main contributions are as follows:

\begin{enumerate}
\item \textbf{Obfuscated access patterns.} We propose a secure access protocol for remote data storage which is able to hide access patterns. Our construction requires at least one of each of the following two types of clients: a \textit{regular client} which stores data, and an \textit{obfuscation client} which hides client's access patterns. Maximum privacy is achieved as long as at least one obfuscation client behaves honestly.

\item \textbf{Concurrent access.} We provide, to our knowledge, the first concurrent-access protocol with access-pattern hiding properties that does not require a trusted third party (e.g. proxy). Our concurrent access protocol is applicable in scenarios with multiple readers, but can cope with having a single writer.

\item \textbf{Small and constant bandwidth.} For all clients with read-write/read-only access, our protocol requires a constant bandwidth that is independent of the size of the store. This is possible because we separate the regular access clients and the security responsible clients (i.e. the obfuscation clients). The bandwidth requirements for interacting with the store are also small. In our current instantiation a single block of data requires a constant two blocks to be transferred.

\item \textbf{Security and performance analysis.} We give a game-based definition of data and access-pattern privacy for \schemeacr{}-like protocols against honest-but-curious storage servers. Furthermore, we apply this new definition to \schemeacr{} and prove it secure. Last but not least we report on the theoretical and observed performance of our protocol thanks to our  proof-of-concept implementation.
\end{enumerate}

\begin{figure*}[t]
	\captionsetup[subfigure]{labelformat=empty}
	\begin{center}
		\begin{subfigure}[t]{.47\linewidth}
			\centering
			\begin{tikzpicture}[every node/.style={node distance=0pt,rectangle,draw,minimum height=20pt,minimum width=30pt,inner sep=5pt}, scale=0.8, transform shape]
			\node [draw=none, text width=20pt, inner sep=0pt] (storeLabel) {\textbf{Store}};		
			\node [right=5pt of storeLabel, fill=light-gray] (b1) {$block_1$};
			\node [right=of b1] (b2) {$block_2$};
			\node [right=of b2, fill=light-gray] (b3) {$block_1$};
			\node [right=of b3] (bd) {$\ldots$};
			\node [right=of bd] (bn) {$block_n$};

			\node [draw=none, above=10pt of b1,minimum height=10pt, inner sep=1pt] (p1) {$p_1$};
			\node [draw=none, above=10pt of b2,minimum height=10pt, inner sep=1pt] (p2) {$p_2$};
			\node [draw=none, above=10pt of b3,minimum height=10pt, inner sep=1pt] (p3) {$p_3$};
			\node [draw=none, above=10pt of bn,minimum height=10pt, inner sep=1pt] (pn) {$p_N$};
			
			\draw[-latex] (p1.south) -- (b1.north);
			\draw[-latex] (p2.south) -- (b2.north);
			\draw[-latex] (p3.south) -- (b3.north);
			\draw[-latex] (pn.south) -- (bn.north);
			
			\node [draw=none, below=15pt of storeLabel.south west, anchor=north west] (blockLabel) {\textbf{Block}};
			\node [right=25pt of blockLabel] (bid) {bid};
			\node [right=of bid] (cns) {cns};
			\node [right=of cns] (ts) {ts};
			\node [right=of ts] (data) {data};
			
			\draw (b3.south west) -- (bid.north west);
			\draw (b3.south east) -- (data.north east);
			
			\node [draw=none, below=50pt of storeLabel] (spacing) {};
			\end{tikzpicture}
			\caption{\small{Figure \ref{fig:serverdata}: Server data structures in \schemeacr{}. The server redundantly stores $n$ equally-sized encrypted blocks at $N$ memory locations, $n<N$. The memory locations are addressable through unique positions ids $p_1,\dots,p_N$. Each $block$ stored contains the data intended for storage (i.e. $block.data$) and a small amount of metadata (i.e. $block.bid$, $block.ts$ and $block.cns$) that helps with map synchronization between concurrent clients.}}
			\label{fig:serverdata}
		\end{subfigure}
		\hfill
		\begin{subfigure}[t]{.47\linewidth}
			\centering
			\begin{tikzpicture}[every node/.style={node distance=0pt,rectangle,draw,minimum height=20pt,minimum width=30pt,inner sep=5pt}, scale=0.8, transform shape]
			\node [draw=none] (mapLabel) {\textbf{Map}};
			\node [right=20pt of mapLabel] (bid1) {$bid_1$};
			\node [right=of bid1] (bid2) {$bid_2$};
			\node [right=of bid2] (bidd) {$\ldots$};
			\node [right=of bidd] (bidn) {$bid_n$};
			
			\node [draw=none, below=15pt of mapLabel.south west, anchor=north west] (bidLabel) {\textbf{BlockID}};
			\node [right=15pt of bidLabel] (psns) {psns};
			\node [right=of psns] (ts_map) {ts};
			\node [right=of ts_map] (vf) {vf};
			
			\draw (bid2.south west) -- (psns.north west);
			\draw (bid2.south east) -- (vf.north east);

			\node [below right=15pt and -40pt of bidLabel] (pa) {$p_a$};
			\node [right=of pa] (pb) {$p_b$};
			\node [right=of pb] (pd) {$\ldots$};
			\node [draw=none, below=of pb] (psnsLabel) {\textbf{Positions}};
			
			\draw (psns.south west) -- (pa.north west);
			\draw (psns.south east) -- (pd.north east);
					
			\node [right=20pt of pd] (vpa) {$p_x$};
			\node [right=of vpa] (vpd) {$\ldots$};
			\node [draw=none, below right=0 and -43pt of vpa] (vfLabel) {\textbf{Verified positions}};

			\draw (vf.south west) -- (vpa.north west);
			\draw (vf.south east) -- (vpd.north east);
			
			\node [draw=none, below right=25pt and 140pt of mapLabel] (spacing) {};
			\end{tikzpicture}
			\caption{Figure \ref{fig:clientdata}: Client local data structures in \schemeacr{}. The client stores a linear map indexed by block ids (i.e. $bid$). For block id $bid$ the client stores a list of server positions (i.e. $bid.psns$) from where the block can be retrieved. The map also keeps some metadata about each block id (i.e. $bid.ts$ and $bid.\vf$) that helps with synchronisation between concurrent clients.}
			\label{fig:clientdata}
		\end{subfigure}
	\end{center}
\end{figure*}
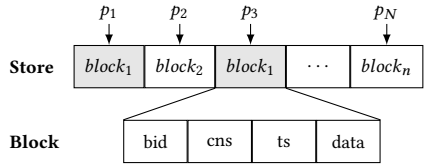

\section{\schemefull{} (\schemeacr{})}
\label{primitive}

\schemeacr{} is a protocol for storing data securely by encrypting it and anonymizing (read or write) access patterns. CAOS allows users to trade storage space and security for concurrency and bandwidth efficiency.

\smallskip\noindent\textit{Data elements.} 
In \schemeacr{} data is partitioned into blocks of equal size. Each block of data is uniquely identified by a client using a \textit{block id (bid)}. Storing a block remotely involves encrypting the contents of the block and then placing the resulting ciphertext at a random location in the store's memory. We refer to these locations at the store's memory as \textit{positions}. The size of the store is measured in the number of positions it has available for storing blocks. Clients can store the same block at multiple positions and keep track of where their data is located by maintaining a \textit{map} which links block ids to positions. 

\smallskip\noindent\textit{Clients.} There are two types of clients in \schemeacr: regular clients (RC), which may be read-write (RW) or read-only (RO);  and obfuscation clients (OC). Both RC and OC access the store directly and independently from each other. 


Regular clients are the main users of the store. They have low bandwidth and local storage requirements, as they only have to store a map. However, accesses done by these clients do leak information about access patterns. 
Obfuscation clients are the clients that provide security. These clients are able to provide access-pattern obfuscation for themselves as well as RCs. For that purpose OCs use of a buffer which is stored locally in addition to a map. The size of the OC's local buffer and the OC's bandwidth requirements are proportional to the speed of obfuscation. 

\smallskip\noindent\textit{Access-pattern obfuscation.} Our definitions derive from existing access-pattern security definitions in ORAM \cite{stefanov2011towards}. Intuitively, the ORAM definitions require that no information should leak with regards to: (1) which data is being accessed, (2) the frequency of accesses, (3) the relation between accesses, (4) whether access is read or write, and (5) the age of the data. ORAM constructions maintain invariants to ensure that no information is leaked regardless of how many times the store is accessed. \schemeacr{} maintains the requirement that no information is leaked for cases (1)-(5), but does not provide guarantees for each individual access operation. Instead, \schemeacr{} provides security guarantees for \textit{access sequences} that involve both regular clients and obfuscation clients. Our security definitions for content and access-pattern security in \schemeacr{} are detailed in Section~\ref{sec:secanalysis}.

\smallskip\noindent\textit{Concurrency.} \schemeacr{} allows multiple clients to access the store simultaneously and independently from each other.
Achieving concurrency in \schemeacr{} is not a trivial task. This is because each client's access operation randomly changes the contents of store, and these changes are only stored locally to that client.
Thus, \schemeacr{} needs to address two problems: (1) to synchronise locally stored client maps in an efficient manner, and (2) how to allow multiple clients to  change the store simultaneously and in a way that does not result in data loss for other clients.

\smallskip\noindent\textit{Syntax.} In the following we draw on the above and give the syntax of our \schemeacr{} protocol. Alternative variants of \schemeacr{} are possible if adhering to this syntax. We say that an $(n, N)$ store $S$ is a collection of $n$ data blocks written to $N$ store positions such that $n<N$.

\begin{definition}
\emph{\schemeacr{}} consists of a tuple of five PT algorithms $\schemesymb{}=\allowbreak (\KGen,\allowbreak\INITSTORE,\allowbreak\INITOC,\allowbreak\ACCESSRW,\allowbreak\ACCESSOC)$ over an $(n, N)$ store $S$: 
	
	\begin{description}
		\item[$k \leftarrow \KGen(1^\lambda):$] is a setup probabilistic  algorithm run by the RW client. It takes as input the security parameter $\lambda$ and outputs a secret key $k$.
		
		\item[$S \leftarrow \INITSTORE(DB,N,k):$] is a deterministic algorithm run by the RW client to initialize the data store. It takes as input a database $DB=(B_0,\ldots,B_{n-1})$ of $n$ data blocks, encrypts each block under key $k$, and distributes them between the total number $N$ of store positions.
		
		\item[$\buf, S \leftarrow \INITOC(S, k):$] is a deterministic algorithm run by the OC to initialize itself. It requires access to an initialized store $S$ and its encryption key $k$ and creates the internal buffer of the obfuscation client.
		
		\item[$ret, S \leftarrow \ACCESSRW(B, op, d, S, k):$] is a probabilistic algorithm that RCs run to access a store. It takes as input the bid $B$ to be accessed, the operation $op \in \{read,write\}$, the data $d$ to be written if $op=write$, and the store $S$ and its key $k$. When the client runs this algorithm, some positions on the server are read, and others are written. It returns the block read or an acknowledgement for the write operation, and the new state of the store $S$.
	
		\item[$\buf, S \leftarrow \ACCESSOC(\buf, S, k):$] is a probabilistic algorithm run by the OC to access a store. It takes as input a local data structure $\buf$ that acts as a buffer, a store $S$ and a key $k$. The algorithm alters the OC's buffer of the obfuscation client. Additionally, when the obfuscation client runs this algorithm, some positions in the store are read, while others are written. This changes the mapping between blocks and positions.
	\end{description}
\end{definition}

\section{Efficient access-pattern obfuscation in \schemeacr{}}
\label{CSW}

This section describes \schemeacr{}. We begin with an overview of the protocol and we will follow up with details about the corresponding algorithms and discussing a proof of concept implementation. The complete source-code is available at \cite{source_code}.

\subsection{Overview}


\smallskip\noindent\textit{Access pattern obfuscation.} In \schemeacr{} we achieve access-pattern obfuscation for sequences of access operations (see Section \ref{sec:secanalysis}). This is a weaker security guarantee than that used in other works \cite{goldreich1987towards, goldreich1996software, stefanov2011towards}, where obfuscation is achieved for each single access operation. In return, our construction allows for concurrency and is more practical.

In \schemeacr{} hiding the type of access (read or write) and the age of the data is done by joining both the read and the write operations into a single access function, $\ACCESSRW$. This prevents the adversary from learning when data is read or written, and when new data is added to the store, with the exception of the initial provisioning of the store done by running $\INITSTORE$.

\schemeacr{} uses a locally stored $map$ per client to keep track of which store positions contain which blocks. By setting the size of the store to be larger than the size of the data to be secured, the algorithm $\ACCESSRW$ can create redundancies through re-encrypting and duplicating blocks from the store and assign them to random free-positions. We use the term \textit{free-positions} to refer both to store positions which have never been written, and to positions whose corresponding blocks have at least one redundancy (i.e. blocks that are stored in two or more places). By allowing regular clients (RCs) to access the same data from multiple store positions we are able to partially obfuscate details about the frequency with which specific data is being accessed, and about the relationship between subsequent accesses. We say partially because, even though data is duplicated to random positions, the adversary can still connect these positions to the initial position from where the duplication process began. To address this issue we use obfuscation clients (OCs). These are read-only clients that use the $\ACCESSOC$ function to access the store similarly to RCs. The difference is that OCs maintain a local buffer which is used to store the contents of the positions received from the store. When an OC performs a store access, it writes (i.e. duplicates) in the store a block read from its buffer. As such, blocks that are duplicated by OCs are not linked to current store blocks and do not leak any access-pattern information.


\smallskip\noindent\textit{Data structures.} 
In \schemeacr{} each store block contains the following: data to be stored \textit{block.data}, a block identifier \textit{block.bid}, a consolidation field \textit{block.cns} that indicates the number of clients that know that a block is stored at a position, and a timestamp \textit{block.ts} of when the data was last changed (cf. Fig \ref{fig:serverdata}).

Client local maps are indexed by the block identifier \textit{bid}, and contain the following: \textit{bid.psns}  enumerates the positions in the store from which the block \textit{bid} is available, \textit{bid.ts} stores most up-to-date timestamp observed, and \textit{bid.\vf} stores positions \textit{p} observed by the map holder (i.e. client) to have \textit{block.cns=$|Clients|$}, where $Clients$ is the set of all clients engaged in our protocol (cf. Fig \ref{fig:clientdata}). 

\smallskip\noindent\textit{Concurrency.} In \schemeacr{} access-pattern obfuscation is achieved through shuffling, thus achieving efficient concurrency with direct access for all clients represents a significant challenge. This is especially difficult because in \schemeacr{} each client maintains its own map and syncs it with the store \textit{independently} from other clients during $\ACCESSRW$ or $\ACCESSOC$ operations. In order to prevent data loss, i.e. that a client looses track of the current data in the store, we ensure that ``for each data block, there exists a valid position that is known to all clients''. Maintaining this invariant has lead to two design constraints: (1) we require that for each single block accessed two positions are read and two positions are written on the server store, and (2) \schemeacr{} can only handle a single read-write client that works concurrently with other read-only and/or obfuscation clients. These restrictions are further discussed in Section \ref{correctness}.

Shared knowledge between clients is tracked using \textit{block.cns} and \textit{bid.\vf}. All clients start from the same version of the map, which is afterwards maintained independently by each one. The protocol requires clients to signal each other when they perform changes to their local maps (i.e. when reassigning a position or when changing the data in a block). Because the client only has access to one position per block during an access operation, the change produced by the client will be localised to that particular position in the store. The problem is that without any additional signalling other clients who are not aware of the change have no way of assessing whether the block stored at a specific position is the correct one (as indicated by their map). 

We indicate shared knowledge about a position as follows.
Whenever a client makes a change to a block,
the value \textit{block.cns} is set to 1, meaning that only one client, the one that made the change is currently aware of the change. 
When other clients access this block they can become aware that a reassignment has taken place by comparing the \textit{block.bid} value stored in the block with the value they were expecting according to their local map. If the values do not match the client will infer that the position used to retrieve the block has been reassigned, and will update their local map accordingly. Similarly, by comparing timestamp data from the local map \textit{map[bid].ts} and from the retrieved block \textit{block.ts} clients can determine if the block's data was updated. 


Once a client becomes aware that a block has been reassigned to a new position $p$ (and has performed changes in its local map) it increments \textit{block.cns} by 1 to signal clients that it is aware of the change. When the \textit{block.cns} value is equal to the number of clients then all the clients can safely assume they have the same view about what block is stored at position $p$. 

\smallskip

Next we continue by specifying \schemeacr{} algorithms.

\subsection {Read-write (and read-only) client access}
\label{access-rw}

\begin{figure}[t]
	\begin{center}
		\begin{tikzpicture}[every node/.style={node distance=0pt,rectangle,draw,minimum height=32pt,minimum width=30pt,inner sep=5pt,align=center}, scale=0.8, transform shape]

		\node (read) {Select\\Position};
		\node [right=7pt of read] (sync) {Sync};
		\node [right=7pt of sync] (write) {Prepare\\Write};
		\node [right=7pt of write] (duplicate) {Duplicate\\Block};
			
		\coordinate[above=10pt of read] (readc);
		\coordinate[above=10pt of sync] (syncc);
		\coordinate[above=10pt of write] (writec);
		\coordinate[above=10pt of duplicate] (duplicatec);
		\draw[latex-] (read.north) -- (readc);
		\draw[latex-] (sync.north) -- (syncc);
		\draw[latex-] (write.north) -- (writec);
		\draw[latex-] (duplicate.north) -- (duplicatec);
		
		\node [minimum height=20pt,anchor=center, above right=20pt and -60pt of write](access) {$\ACCESSRW$};
		\coordinate[below=10pt of access] (accessc);
		\draw (access.south) -- (accessc);
		
		\draw (accessc) -- (readc);
		\draw (accessc) -- (syncc);
		\draw (accessc) -- (writec);
		\draw (accessc) -- (duplicatec);
			
		\end{tikzpicture}
		\caption{\small{Client read-write access function. The function $\ACCESSRW$ performs four actions: (1) selects the position used to retrieve a block from the store, (2) synchronises local map with metadata from retrieved blocks, (3) prepares a block to be written back by ensuring that the write operation is possible, and (4) attempts to duplicate one of the retrieved blocks onto the position of the other.}}
		\label{fig:accessmethod}
	\end{center}
	\vspace{-15pt}
\end{figure}
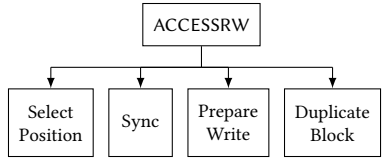

Read-write (RW) clients, such as  email SMTP servers, and read-only (RO) clients, such as email readers, perform  reading and writing through the function $\ACCESSRW$  (cf. Fig \ref{fig:accessmethod}).
We introduce the main function first with calls to several sub-functions that implement required functionality. Each sub-function is subsequently described.

\begin{algorithm2e}[h]
\small
\SetKwProg{function}{function}{}{}
\SetKwProg{rnd}{in random order}{}{}

\KwIn{block id, operation, data (for write operation), store, client $map_c$, store key}
\KwOut{data (for read operation), store, client $map_c$}

\SetKwProg{fn}{function}{}{}

\BlankLine
\function{$\ACCESSRW$ (bid, op, data, S,  $map_c$, k)}{	
	\BlankLine
	$req\_p \gets SelectPosition (bid, map_c)$\; \label{a1l1}
	$cpy\_p \gets SelectPosition (null, map_c)$\; \label{a1l2}
	
	\BlankLine
	$(req\_blk,cpy\_blk) \gets$ READ($req\_p,cpy\_p,S,k$)\; \label{a1l3}
	
	\BlankLine
	$(req\_blk,map_c) \gets Sync(req\_blk, req\_p, map_c)$\; \label{a1l4}
	$(cpy\_blk,map_c) \gets Sync(cpy\_blk, cpy\_p, map_c)$\; \label{a1l5}
	
	\BlankLine
	\If {$(op = READ)$}{\label{a1lif}
		\If {$(req\_blk.bid = bid)$}{
			$out \gets req\_blk.data$\;
		}\label{a1lif1}
	}
	\If {$(op = WRITE)$}{\label{a1lel}
		$(req\_blk,out) \gets PrepareWrite (bid, req\_blk, data)$\;
	}\label{a1life}
	
	\BlankLine
	$(cpy\_blk, map_c) \gets DuplicateBlock(req\_blk, cpy\_blk, cpy\_p, map_c)$\;\label{a1dup}
	
	\BlankLine
	$S\gets$WRITE($S,k,req\_blk, cpy\_blk, req\_p, cpy\_p$)\;\label{a1wrt}
	
	\BlankLine
	\Return $(out, S, map_c)$\;\label{a1ret}
	
}
\caption{Main \schemeacr{} access function.}
\label{algo:csw-main-rw}
\end{algorithm2e}

The main $\ACCESSRW$ function detailed in Algorithm \ref{algo:csw-main-rw} involves the following 8 main steps: 
\begin{enumerate}    
    \item \textit{Select positions} (line \ref{a1l1},\ref{a1l2}). First, select a position (line \ref{a1l1}) to read/write a block as indicated by the input value $bid$. 
    Then select a random position to sync to the local map (line \ref{a1l2}). The block requested via $bid$ will also be copied to the random position if necessary conditions are met (see Algorithm \ref{algo:csw-duplicate}).
    
    The $SelectPosition()$ function interacts with the locally stored map data structure and selects a random position for a given block id $bid$. If no $bid$ is supplied a completely random position is chosen from the set of known store positions. 
    
    \item \textit{Read positions from store} (line \ref{a1l3}). Retrieve the block indicated by $bid$ and the random block using previously selected positions and decrypt them.
  
    \item \textit{Sync} (lines \ref{a1l4},\ref{a1l5}). Synchronize the metadata from the blocks retrieved with the local client map.
    
    \item \textit{Return read data} (lines \ref{a1lif}-\ref{a1lif1}). Prepare to return data if the operation is $READ$ and the block retrieved is correct. If the block retrieved has a different $bid$ than the one requested the client has to re-run $\ACCESSRW$. The protocol guarantees that every client has a valid position for every block.
    
    \item \textit{Replace block data} (lines \ref{a1lel},\ref{a1life}). If the operation is $WRITE$ replace the current data in $req\_blk$ with the input data $data$. 
    
    \item \textit{Duplicate block} (line \ref{a1dup}). Attempt to shuffle store data by coping the $req\_blk$ to position $cpy\_p$, thus increasing the number of positions for $req\_blk$ and reducing the available positions for block $cpy\_blk$ (see Algorithm \ref{algo:csw-duplicate}). 
    
    \item \textit{Write blocks to store} (line \ref{a1wrt}). Encrypt $req\_blk$ and $cpy\_blk$ and write them at positions $req\_p$ and $cpy\_p$.  
    
    \item \textit{Return} (line \ref{a1ret}). Return $out$ as requested data if $op=READ$ and operation was successful; return $null$ otherwise. Return input $data$ as $out$ if $op=WRITE$ and operation was successful ($null$ otherwise). 
\end{enumerate}

\subsubsection {Sync function}
\label{csw-sync}

\begin{algorithm2e}
\small
\SetKwProg{function}{function}{}{}

\KwIn{block, position, client map}
\KwOut{block, client map}

\SetKwProg{fn}{function}{}{}

\BlankLine
\function{Sync (block, p, $map_c$)}{

	\If {$(block.ts < map_c[block.bid].ts)$}{
		remove $p$ from $map_c[block.bid].psns$ and $map_c[block.bid].vf$\;
		$block.bid\gets free$\;
		$block.ts\gets current\_time$\;
		$block.cns \gets 1$\;
	}\Else{
		\If {$(block.cns < \vert clients\vert)$} {
			\If {$(block.ts > map_c[block.bid].ts)$}{
				clear $map_c[block.bid].psns$ and $map_c[block.bid].vf$\;
				set $map_c[block.bid].ts$ to $current\_time$\;
			}
			
			\If {$(p \notin map_c[block.bid].psns)$}{
				move $p$ to $map_c[block.bid].psns$\;
				$block.cns \gets block.cns+1$\;
			}
		}
		\If {$(block.cns = \vert clients\vert)$}{
			add $p$ to $map_c[block.bid].vf$\;
		}
	}
	\Return (block, $map_c$)\;
}

\caption{Sync metadata  between a store block and the local map.}
\label{algo:csw-sync}
\end{algorithm2e}

This function (see Algorithm \ref{algo:csw-sync}) enables \schemeacr{} clients to work concurrently, without having to synchronize their maps to access the store. We do require that once, during the setup phase, clients share the store key and a map from one of the other participating clients. Once a copy of the map has been obtained each client can maintain its own version of the map i.e. $map_c$.

Intuitively this works as follows: after a block has been retrieved from the store, the client needs to establish if local knowledge about the block is correct, i.e. $bid$ of the retrieved block and the position used to retrieve it are correctly linked, and the time-stamp of the block is not older than the time-stamp stored locally. As such, the block can be in multiple states of which only the following require a map modification: (1) old data -- the block is marked as free; (2) wrong position, data up-to-date -- update bid/position association; (3) new data -- update bid/position association and local time-stamp.

Old data in a block is detected by comparing the $block.ts$ value with the locally stored $map_c[block.bid].ts$ value, where $block.bid$ is the bid of the block being synchronised. Old data is thus easy to detect as each $bid$ entry from the map has a single time-stamp (i.e. $map_c[block.bid].ts$): the newest time-stamp observed over all store accesses. In this case the position used to retrieve the block can be used for writing or duplication. 
If new data is detected, the local time-stamp is updated and positions pointing to the old version of the block are removed from $map_c[block.bid].psns$ and $map_c[block.bid].\vf$. They can now be reused for writing or duplication.

If a wrong position is detected then $map_c$ is updated with correct bid-position association (the position thought to point to the requested $bid$ is moved to the correct $map_c[block.bid].psns$). Any modification to a client's local map is communicated to other by incrementing $block.cns$.

\subsubsection{PrepareWrite function}
\label{csw-write}

Writing or retrieving a block to/from the store is not guaranteed to succeed. The PrepareWrite function ensures that all necessary conditions are met before overwriting the data of a block. A block can only be replaced if the id of the block requested matches the id of the block fetched.

In the event that the block fetched is different than the block requested data can still be replaced but only if the data in the fetched block was old.  Replacement of data inside the block is indicated to other clients by resetting the block.cns value to 1 (i.e. only the current client knows about it) and updating the block.ts to current epoch time. We give the pseudo-code for this function in Appendix \ref{app:algorithms}, while Algorithm \ref{algo:csw-write} presents the PrepareWrite function in detail.

\subsubsection {DuplicateBlock}
\label{csw-duplicate}

As previously mentioned, data shuffling is performed through duplication, with two purposes: (1) increasing the number of available positions for blocks accessed more often; (2) reducing the available positions for the less accessed blocks.

The \textit{DuplicateBlock} operation performs the necessary checks and attempts to reassign positions from one block to another while ensuring that clients still have at least one position for each block. 
The difficulty of this operation resides in the fact that each client has a slightly different version of the map and can only become aware of changes after they have happened (often with significant delay).
In order to allow clients to attain a shared view about a position we require that positions do not get reassigned while their block.cns value is below the maximum number of clients. The single exception to this rule is when a client has reassigned a position, decides to reassign it again and no other client has noticed the reassignment. This is allowed because it will not affect the convergence time for that position. The specifics are presented in Appendix \ref{app:algorithms}, Algorithm \ref{algo:csw-duplicate}.

\subsection {Obfuscation client access}
\label{csw-oc}

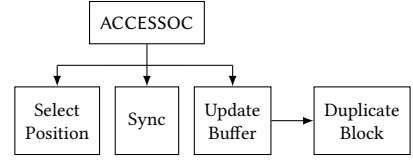
\begin{figure}[t]
	\begin{center}
		\begin{tikzpicture}[every node/.style={node distance=0pt,rectangle,draw,minimum height=32pt,minimum width=30pt,inner sep=5pt,align=center}, scale=0.8, transform shape]
		
		\node (read) {Select\\Position};
		\node [right=7pt of read] (sync) {Sync};
		\node [right=7pt of sync] (update) {Update\\Buffer};
		
		\node [right=20pt of update] (duplicate) {Duplicate\\Block};
		\draw[-latex] (update.east) -- (duplicate.west);
				
		\node [minimum height=20pt,anchor=center, above = 20pt of sync](access) {$\ACCESSOC$};
		\coordinate[below=10pt of access] (accessc);
		\draw (access.south) -- (accessc);
		
		\coordinate[above=10pt of read] (readc);
		\coordinate[above=10pt of sync] (syncc);
		\coordinate[above=10pt of update] (updatec);

		\draw[latex-] (read.north) -- (readc);
		\draw[latex-] (sync.north) -- (syncc);
		\draw[latex-] (update.north) -- (updatec);
				
		\draw (accessc) -- (readc);
		\draw (accessc) -- (syncc);
		\draw (accessc) -- (updatec);
		
		\end{tikzpicture}
		\caption{\small{ $\ACCESSOC$ performs three actions: (1) selects and retrieves two blocks from the store, (2) synchronizes local map with metadata from retrieved blocks, and (3) attempts to replace retrieved blocks with buffered stored blocks.}}
		\label{fig:access_oc}
	\end{center}
	\vspace{-10pt}
\end{figure}

In \schemeacr{} data shuffling is performed on a single position at a time. This process leaks a lot of information to the store server because every run of the $\ACCESSRW$ function links two positions together with high probability (i.e. ideally, we require that the $DuplicateBlock$ function is successful every time).

To mitigate this we use a separate obfuscation client  function $\ACCESSOC$ as presented in Algorithm \ref{algo:csw-main-oc}. This function is identical to the $\ACCESSRW$ function as far as the interaction with the store is concerned: two blocks are read from two server positions and two blocks are written to the same positions.

\begin{algorithm2e}[h]
\small
\SetKwProg{function}{function}{}{}
\SetKwProg{rnd}{in random order}{}{}

\KwIn{local buffer, obfuscation client map, store, key}
\KwOut{store, obfuscation client map}

\SetKwProg{fn}{function}{}{}

\BlankLine
\function{$\ACCESSOC$ (buf, S, $map_{oc}$, k)}{
	
	\BlankLine
	$p_1 \gets SelectPosition (null, map_{oc})$\;
	$p_2 \gets SelectPosition (null, map_{oc})$\;
	
	\BlankLine
	$(blk_1,blk_2) \gets$ READ($p_1,p_2,S,k$)\;
	
	\BlankLine
	$(blk_1,map_{oc})\gets Sync(blk_1, p_1, map_{oc})$\;
	$(blk_2,map_{oc})\gets Sync(blk_2, p_2, map_{oc})$\;
	
	\BlankLine
	$(blk_1, buf) \gets \mathit{UpdateBuffer}(blk_1, p_1,\mathit{buf})$\;
	$(blk_2, buf) \gets \mathit{UpdateBuffer}(blk_2, p_2,\mathit{buf})$\;

	\BlankLine
	$S\gets$WRITE($S,k,blk_1,blk_2,p_1,p_2$)\;
	
	\Return (S,$map_{oc}$)\;
}
\caption{Obfuscation client main access function.}
\label{algo:csw-main-oc}
\end{algorithm2e}

The purpose of the OC is to prevent the store from linking the positions that are duplicated thorough the regular $\ACCESSRW$ function. This is done through the use of an OC locally stored buffer. $\ACCESSOC$ places the blocks read from the store into its buffer, while blocks chosen at random from those currently stored in the buffer are written back to store. The replacement procedure is bound to the same constraints as when blocks are duplicated by regular clients, so we use the same DuplicateBlock functionality to ensure these are enforced. 
A detailed version of our method is presented in Algorithm \ref{algo:csw-main-oc}. 


\section{Concurrency and parallel access in CAOS}
\label{sec:conc-proto}

In \schemeacr{} concurrency is supported by design. Regular clients communicate with the store using two messages that are abstracted here as follows: $READ(p_1,p_2)$ and $WRITE(p_1,p_2)$. The message $READ(p_1,p_2)$ is used to request the two blocks located at positions $p_1$ and $p_2$ in the store.  Upon receiving this message the server will lock $p_1$ and $p_2$ and will return the corresponding blocks to the client. The client will process the data, and will use the blocks metadata to update its local map (cf. Section \ref{CSW}). While the client processes the contents of the blocks the server will keep $p_1$ and $p_2$ locked. The positions $p_1$ and $p_2$ are unlocked when the server receives a $WRITE(p_1,p_2)$ message from the client that locked the positions. In the event the client is unable to send a $WRITE(p_1,p_2)$ message a timeout period is used to unlock them.

Recall that the client is  interested in retrieving the content of only one of the positions, say $p_1$, where $p_1$ is chosen randomly from the set of positions associated to a given block id in the client's local map. The other position $p_2$ is chosen from the store at random and is used to duplicate the data from $p_1$. As such, a second client can request the same block using different positions. In the event that a second client requests the same positions $p_1$ and/or $p_2$ while they are locked, the server will inform the client of the lockout through an error message. By running $\ACCESSRW$ again, the second client can select new positions $p'_1$ and $p'_2$ to retrieve the requested block and a duplicate-destination block. A detailed instantiation of this procedure is presented in Fig \ref{fig:concurrency}.

A secondary benefit from supporting concurrency by design is that it allows \textit{parallel access} to the store. Because each store access by CAOS functions only affects two positions at a time, regular clients (and OCs) can instantiate multiple $\ACCESSRW$ (or $\ACCESSOC$) operations simultaneously for different blocks. This increases access speeds when the data required is stored in multiple blocks. It also affects the speed of the obfuscation process which can be increased or decreased by adjusting the bandwidth available.

\section{Content and Access-Pattern Privacy for \schemeacr{} and Protocol Correctness}
\label{sec:secanalysis}

In this section we give a formal definition of content privacy and pattern access privacy for \schemeacr-like protocols. Next we analyse the privacy offered by our proposal \schemeacr{}. Finally, we present our invariants which ensure that our asynchronous  concurrency protocol does not result in unintentional data loss.

\begin{figure}[t]
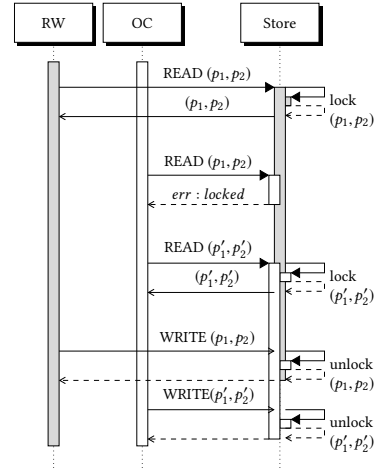

   \centering
   \scalebox{0.65}{
		\begin{sequencediagram}
	        \pgfumlsdunderlinefalse
	        \newthread{rw}{RW}
	        \newthread[white]{oc}{OC}
	        \newinst[1.2]{s}{Store}
	
	        \begin{call}{rw}{\hspace{45pt} READ $(p_1,p_2)$}{s}{}
	         \prelevel
	        	\begin{callself}{s}{}{\hspace{-5pt}\begin{tabular}{l}lock\\$(p_1,p_2)$\end{tabular}}
	        	\end{callself}
	        	\prelevel
	        	\mess{s}{$(p_1,p_2{)}$}{rw}{45pt}
	        	\postlevel
	        	\postlevel
	        	\postlevel
	        	\postlevel
	        	\postlevel
	        	\postlevel
	        	\postlevel
	        	\postlevel
	        \end{call}
	        
	        \prelevel
	        \prelevel
	        \prelevel
	        \prelevel
	        \prelevel
	        \prelevel
	        \prelevel
	        \prelevel
	        \setthreadbias{west}
	        
	        \begin{call}{oc}{READ $({p}_1,{p}_2)$}{s}{$err:locked$}
	        \end{call}
	        \postlevel
	        
	        \begin{call}{oc}{READ $(p_1',p_2')$}{s}{}
	        \prelevel
	        	\begin{callself}{s}{}{\hspace{-5pt}\begin{tabular}{l}lock\\$(p_1',p_2')$\end{tabular}}
	        	\end{callself}
	        	\prelevel
	        	\mess{s}{$(p_1',p_2'{)}$}{oc}{5pt}
	        	\postlevel
				\mess{rw}{WRITE $({p}_1,{p}_2{)}$}{s}{45pt}
	        	\prelevel
	        	\begin{callself}{s}{}{\hspace{-5pt}\begin{tabular}{l}unlock\\$({p}_1,{p}_2)$\end{tabular}}
	        	\end{callself}
	
	        	\mess{oc}{WRITE$(p_1',p_2'{)}$}{s}{-5pt}
	        	\prelevel
	         \begin{callself}{s}{}{\hspace{-5pt}\begin{tabular}{l}unlock\\$(p_1',p_2')$\end{tabular}}
	         \end{callself}
				\prelevel
	        \end{call}
		\end{sequencediagram}
	}
	\caption{\small{Concurrent access in \schemeacr{}. Upon receiving a $READ(p_1,p_2)$ message the server locks requested positions $p_1$ and $p_2$. If receiving a second request for $p_1$ and/or $p_2$, from a different client (e.g. the OC), the server will reply with an error. The second client can restart the protocol and request a new pair (e.g. $p'_1$ and $p'_2$). Positions are unlocked after the server processes the $WRITE$ position message (e.g. $WRITE(p_1,p_2)$ or $WRITE(p'_1,p'_2)$).}}
	\label{fig:concurrency}
	\vspace{-10pt}
\end{figure}

\vspace{10pt}

\subsection{Game-based privacy for \schemeacr{}}

We start with the following definition.


\begin{definition}
The \emph{access-pattern} induced by the $i$-th run of the $\ACCESSRW$ or $\ACCESSOC$ algorithms of a \schemeacr-like protocol $\schemesymb$  is the tuple $AP_i=(p_{r_1},\ldots,p_{r_l},\allowbreak p_{w_1},\allowbreak\ldots,p_{w_m})$, consisting of the $l$ positions read by the server and the $m$ positions written by the server. The access-pattern induced by a sequence of queries is the combination of the patterns induced by the individual queries.
\end{definition}

\begin{model} When the algorithms $\ACCESSRW$ and $\ACCESSOC$ are run,
  the adversary learns the access-pattern induced by those
  queries. The attacker is given access to the setup algorithms
    $\INITSTORE$ and $\INITOC$, and thus knows the initial layout of
    the store and OC buffer.
\end{model}

Using the access-pattern definition above we define the security of a generic \schemeacr-like construction $\schemesymb$ against a \textit{multiple query attacker} in the following.

\begin{definition}[data and access-pattern privacy]
	\label{def:multi-q}
	\schemeacr{} \emph{data and access-pattern privacy} is defined through a multiple query security experiment as follows (see also Fig \ref{fig:sec_exp1}).
	Let $\schemesymb=\allowbreak  (\KGen,\allowbreak\INITSTORE,\allowbreak\INITOC,\allowbreak\ACCESSRW,\ACCESSOC)$ be a \schemeacr-like protocol over an $(n, N)$ store $S$,  $\lambda$ a security parameter, $r$ be number of rounds the OC is run, and $b\in\{0,1\}$. Let $\adv=(\adv_{1,1},\ldots,\adv_{1,q},\adv_2)$  be a $q$-query adversary. Let $DB$ be a database. We consider ${\Exp}_{m,\adv}^{b}(\schemesymb)$, a probabilistic experiment defined in terms of a game played between an adversary $\adv$ and a challenger $\cdv$, consisting of:
\end{definition}

\begin{enumerate}[leftmargin=15pt]
	\item \textit{Setup.} 
	$\cdv$ runs $\KGen(1^\lambda)$ to create a symmetric key $k$, and runs $(\INITSTORE,\allowbreak \INITOC)$ to initialize the  store and the OC's internal buffer.
	
	\item{Query.} During each query $j$:
		\begin{enumerate}[leftmargin=5pt]
		\item \textit{Obfuscation.} $\cdv$ runs the obfuscation algorithm $\ACCESSOC(S,\allowbreak k)$, $r$ times.
	
		\item \textit{Challenge.} $\adv_1$ chooses two block id's, $B_0$ and $B_1$, two operations $op_0$ and $op_1$, and two data contents $d_0,d_1$ to be written if any of $op_0$ or $op_1$ is a write operation. $\cdv$ runs $\ACCESSRW(B_b, \allowbreak{op_b},\allowbreak{d_b}, \allowbreak S, \allowbreak k)$.
	
		\end{enumerate}
	
	\item \textit{Guess.} $\adv_2$ computes a guess $b' \in \{0,1\}$,  winning the game if $b'=b$.
 The output of the experiment is $b'$.
\end{enumerate}

The \emph{advantage of a multiple query adversary} $\adv=(\adv_{1,1},\ldots,\allowbreak\adv_{1,q},\allowbreak\adv_2)$ against the \emph{data and access-pattern privacy} of a \schemeacr-like protocol $\schemesymb$ is defined as: $$\displaystyle\big\lvert\prob{{\Exp}_{m,\adv}^{0}(\schemesymb)=1}- \prob{{\Exp}_{m,\adv}^{1}(\schemesymb)=1}\big\rvert.$$

\begin{figure}[t!]
	\hspace{-8pt}
	\framebox[0.5\textwidth][t]{
		\small
		\hspace{-5pt}
		\begin{tabular}{m{0.47\textwidth}}
			$Exp_{\mathsf{m},\mathcal{A}}^{b}(\schemesymb)$\\[2pt]\hline
	
			$k \xleftarrow{R} \KGen(1^\lambda)$ \\[2pt]
			$st_{0} \gets (\INITSTORE, \INITOC)$\\[2pt]
			$\mathbf{for}~j \in \{1,\ldots,q\}~\mathbf{do}$\\[2pt]
			\enskip $\mathbf{for}~ i \in \{0,\ldots,r-1\}~\mathbf{do}$\\[2pt]
			\enskip\enskip $st_{i} \gets \ACCESSOC(S,k)$\\[2pt]
			\enskip $\mathbf{endfor}$ \\[2pt]
			\enskip $(aux_j,B_{0,j}, B_{1,j}, op_{0,j}, op_{1,j}, d_{0,j},d_{1,j})\gets \mathcal{A}_{1,j}(st_{0},\ldots,st_{r-1})$ \\[2pt]
			\enskip $st_{r,j} \gets \ACCESSRW(B_{b,j}, op_{b,j}, d_{b,j}, S, k)$ \\[2pt]	
			$\mathbf{endfor}$ \\[2pt]
			$b' \gets \mathcal{A}_{2}(aux_1,\ldots,aux_q,{st_{r,1}},\ldots,{st_{r,q}})$\\[2pt]
			$\mathbf{return}~ b'$
		\end{tabular}
	}
	\caption{\small{\schemeacr{} data and access-pattern privacy game for a multiple query adversary $\mathcal{A}=(\mathcal{A}_{1,1},\ldots,\mathcal{A}_{1,q},\mathcal{A}_2)$.}}
	\label{fig:sec_exp1}
\end{figure}

Let us argue that the security experiment above captures both content and access-pattern privacy for \schemeacr-like protocols. Data privacy against the server is captured by letting the adversary choose two equal-sized data blocks $B_0$ and $B_1$ (indexed by their block id's) before every $\ACCESSRW$ query. The read-write (read-only) client will call $\ACCESSRW$ on data block $B_b$, where $b\in\{0,1\}$ is unknown to the adversary. On the other hand, access-pattern privacy against the server is captured by an adversarial choosing of the operation to be performed, i.e. either $op_0$ or $op_1$, depending on the value of $b$ that is unknown to the adversary. During the experiment, the adversary $\adv$ has full knowledge of the instructions run and the changes occurred in the store when the read-only (read-write) and obfuscation clients will be calling the algorithms $\ACCESSRW,\ACCESSOC$ as defined by the corresponding protocol $\schemesymb$. The attacker's advantage as defined above measures how well the adversary does in gaining knowledge on $b$ by querying $\ACCESSRW,\ACCESSOC$ a number $j$ of times.
We obtain the following result:

\begin{theorem} \schemeacr{} has content and access-pattern privacy, i.e. the advantage of any multiple query adversary against the privacy of $\schemesymb$ is negligible in the security parameter $\lambda$ and the number of OC rounds $r$. 
\end{theorem}

See Appendix \ref{app:privacy} for a proof.

\subsection {Invariants in \schemeacr{}}
\label{correctness}

Different clients hold different position maps, and therefore it is possible that a client's position map will be out of date. Being a little bit out of date is not a problem: if a client seeks a block at a position and finds a different block there, the client can detect that, and seek the block at another position. However, we need to ensure that a client will always eventually find the block. We therefore prove the
following invariant: \emph{for every client and every block, the client has
a valid position for the block in its map}. See Appendix \ref{app:invariants} for a proof.

\section{Performance and Implementation}

The security of \schemeacr{} derives from the fact that the obfuscation client (OC)  runs in between the accesses made by the regular clients. Intuitively, the more runs of OC, the more secure the system. In this section, we derive the expected number of rounds of OC needed to get perfect security (i.e.\ a situation in which the adversary has no knowledge of what block is in what position). In our algorithm, OC processes two positions per round. For simplicity, in this section, we consider an OC that processes a single position per round. Intuitively, the performance of the two-round OC is no worse than a one-round OC.

\begin{theorem}\label{thm:expected}
	The expected number of rounds of OC needed to obtain perfect privacy for a $(n,N)$ store is at most $2+s+ (n-s)\log(n-s)+ N\log N$, where $s$ is the OC's buffer size.
\end{theorem}

 Since $N$ is much greater than $n$ and $s$, this is dominated by $N\log N$. Intuitively, perfect privacy cannot be achieved with less than $N$ rounds, since every store position needs to be updated. See Appendix \ref{app:perfanalysis} for a proof.

\subsection{Performance}

We next report on the performance of our instantiation of \schemeacr{} in terms of storage and bandwidth, and present measured data from a non-optimized implementation of our protocol \cite{source_code} for a 64GB store size.

\smallskip\noindent\textit{Read-write client.}
Client storage in \schemeacr{} is limited to the client map and the two blocks from the store retrieved by the $\ACCESSRW$. These, however, are only stored only for the duration of the method run.

%
Our bandwidth requirements are constant: two blocks are accessed for each requested data block from the store. Because the storage server only locks the positions requested for the duration of a client access, multiple blocks can be requested simultaneously in a parallel manner.

\smallskip \noindent \textit{Obfuscation client.} The OC accesses the store the same way the RW client does, hence the bandwidth requirements per block are the same. The OC requires additional storage than the read-write client in the form of a local buffer which is needed for the obfuscation process. In Appendix \ref{app:perfanalysis} we show how to compute $r_1$ the number of rounds required to obtain secure output from the OC for any given buffer size $s$. We note that the rate given by Eq. \ref{eq:r1adv}, Appendix \ref{app:perfanalysis} is valid for a single positions processed per round. Given that the OC behaves identical to the read-write client implies that the rate is in fact  doubled.

Similarly, one can also compute an upper bound on the number of rounds needed to bring the store from a completely insecure state (i.e. adversary knows the contents of each position) to a secure state using Eq. \ref{eq:r2exp}, Appendix \ref{app:perfanalysis}.

\smallskip \noindent \textit{Storage server.}
In order for our \schemeacr{} protocol to guarantee that blocks can be shuffled between positions it is required that each block is stored redundantly on a number of positions equal to the number of clients using \schemeacr{}. However, if the additional space is not available when the store is initialised, our protocol will automatically adjust itself allowing the size of the store to be increased gradually.

\smallskip \noindent \textit{Implementation.}
We have implemented \schemeacr{} in C++ using the OpenSSL-1.0.2.k for the encryption operations (concretely we use AES-128 in CBC mode) and Protocol Buffers\footnote{https://developers.google.com/protocol-buffers} for network message serialisation. Our setup consists of a server which exposes a block-store API with direct access to disk. The server's API is addressable through positions. In addition to the server we implemented a RW client and an OC client which are able to run $\ACCESSRW$ and $\ACCESSOC$ respectively.

On an Intel i5-3570 CPU at 3.40GHz running ArchLinux we have instantiated a secure store of 64GB with a redundancy factor $C=2$ (i.e. every block is stored on two positions resulting in $N=2n$). We measured the read/write speed to the store, client storage and server storage. Our results are presented in Table \ref{tbl:performance}.

\begin{table}[t]
  \centering
  \begin{tabularx}{0.46\textwidth}{|>{\hsize=0.38\hsize}X|>{\hsize=0.33\hsize}X|>{\hsize=0.29\hsize}X|}
  	 	\hline
    							& Protocol 	& Instantiation\\
    		& $\schemesymb$ Sct. \ref{CSW} &	with  $n=10^6$  			\\	\hline\hline
    	Store capacity		& $n\cdot |block|$	& $64$ GB \\\hline
    	Server storage	&  $C\cdot n \cdot |block|,$ & $128$ GB\\
    	requirements & $C > 1$ & $C=2$ \\\hline\hline
    	\begin{tabular}[x]{@{}l@{}} RW/RO client sto-\\rage requirements\end{tabular}		& $c\cdot n$ ($c \ll 1$)& 55 MB \\\hline
    	\begin{tabular}[x]{@{}l@{}} RW/RO client\\bandwidth\end{tabular} & $ps_1 \cdot 2\cdot |block|$ & $1 \cdot 2\cdot 64$KB\\\hline\hline
    	OC buffer & $s\cdot |block| $& 655 MB \\
    	size			   	& ($2\leq s \leq n$)& ($s=10^4$) \\\hline
    	\begin{tabular}[x]{@{}l@{}} OC\\bandwidth\end{tabular} & $ps_2 \cdot 2\cdot |block|$ & $1 \cdot 2\cdot 64$KB\\\hline
  \end{tabularx}
  \flushright
  \begin{tabularx}{0.286\textwidth}{|>{\hsize=0.53\hsize}X|>{\hsize=0.47\hsize}X|}\hline
      	$\frac{\text{client~storage}}{\text{store~capacity}}$ &  0.08\% \\\hline
      	$|block|$ &  64KB \\\hline
      	Measured &131 KB/s\\
      	speed		& per thread\\\hline
  \end{tabularx}
  \hspace{1pt}
  \vspace{5pt}
  \caption{\small{Storage and bandwidth requirements to store of $n$ data blocks of size $|block|$ bytes in our \schemeacr{} instantiation.  $c,C$ are constants that are implementation-dependent and $ps_1$ and $ps_2$ are the number of parallel sessions run by the client and by the OC respectively. Figures are given in a generic notation (middle column), and as obtained from our own non-optimized implementation of the protocol (right column).}}
  \label{tbl:performance}
  \vspace{-10pt}
\end{table}

\section{Conclusion}
In this paper we have proposed \schemefull{} (\schemeacr{}), a  cloud storage solution that is able to provide access-pattern obfuscation. We have presented our concurrent access protocol that allows multiple read-only clients to simultaneously access a \schemeacr{} store. We have proved that the concurrent access does not result in data loss. We have also proven the security of our protocol for any buffer size $s$ used by the obfuscation client given a sufficient numbers of rounds $r$. Finally, we have shown that security provided by the \schemeacr{} protocol is proportional to the resources provisioned for access-pattern obfuscation, namely the buffer size $s$ and the number of rounds $r$. 
To our knowledge, there is no existing work that uses an obfuscation client to hide access patterns. However several works have been influential in the development of \schemeacr{}, and we list them in Appendix \ref{app:related}.

\bibliographystyle{ACM-Reference-Format}
\bibliography{caos}

\appendix
\section{\schemeacr{} use-case example}
Alice wishes to move her 1TB encrypted email store to the cloud, for easy access from her three devices: a work computer, a laptop and a mobile device. She also requires that her devices have simultaneously access and access-pattern protection.


One option for Alice is to use an ORAM scheme \cite{stefanov2013path, ren2015constants, devadas2016onion} as they provide both data encryption and access pattern security with $O(\log N)$ efficiency. 
With ORAM, Alice can expect to download and upload an average of about 468KB of data\footnote{For a 1TB ORAM store which uses 4KB blocks and 4 blocks per bucket, the size of transferred data is between 104KB and 832KB, depending on how the data is laid out in the server's memory and the ORAM scheme used.} for each 30KB email message she wants to access. 

To enable simultaneous access for her devices she can use a proxy-based ORAM scheme \cite{stefanov2013oblivistore, sahin2016taostore}. This would allow Alice's devices to use less bandwidth, because they would interact with the proxy, however, the bandwidth between the proxy and the store will remain similar (i.e. \cite{stefanov2013oblivistore} requires 30MB/s to provide 1MB/s access to the store, \cite{sahin2016taostore} requires a transfer of 256KB to access a 4KB store block). 
Unfortunately, the high bandwidth and computational requirements might force Alice to place the proxy in the cloud, an option she does not like because then access patterns to the proxy would leak information.

An alternative option for Alice would be to use \schemeacr{}. In \schemeacr{} her devices would have simultaneous and direct access to the store. Access-pattern obfuscation would be done by OCs which can be placed either in the cloud or on the premises because they access the store independently from her devices. However, access operations in \schemeacr{} randomize the store, and Alice's devices and OCs would require some form of synchronisation between them to be able to access it.

One na\"ive \schemeacr{} implementation which solves the synchronisation problem is to separate a part of the store, share it between devices and OCs, and use it to log the changes. During each access, clients would be able to process the log and update their access information (i.e. maps). However, if Alice were to stop using any one of her 3 device for just a month, upon resuming use, that device will have to process around 46MB worth of log entries\footnote{This number assumes that a log entry is 1KB, and that the remaining 2 devices will only access a 30KB email index that uses 8 blocks, 6 times per hour, 16 hours per day for 30 days. The number does not account for any OCs doing access-pattern obfuscation, in which case the number would be much higher.}
only to access a 30KB index, as the log grows linearly with the number of accesses. Security would also be affected. At the very least the variable network traffic required to retrieve the log reveals the period when the device was offline.

\section{Performance analysis}
\label{app:perfanalysis}

In the following we give a proof for Theorem \ref{thm:expected}.

\begin{proof}
	The number of rounds of OC needed is the sum of:
	\begin{enumerate}
		\item The number $r_1$ of rounds needed for the buffer contents to be randomised (this is needed only the first time OC is run, when the adversary knows the contents of the buffer);
		\item The number $r_2$ of rounds needed for every position in the store to have been overwritten by the buffer.
	\end{enumerate}
	To calculate $r_1$, consider an array of size $n$ containing all the blocks.
	Let $U$ be the sub-array containing the first $s$ blocks of the array (these represent the contents of the OC buffer), and $V$ the remaining $n-s$ blocks (representing those that are not in the buffer). In each round of OC, a random block $B$ from $U$ is selected and moved to $V$ (this corresponds to evicting a block from the buffer), while a random block from $V$ is moved to the former place of $B$ in $U$. We seek the expected number $r_1$ of runs which completely randomises the array.
	
	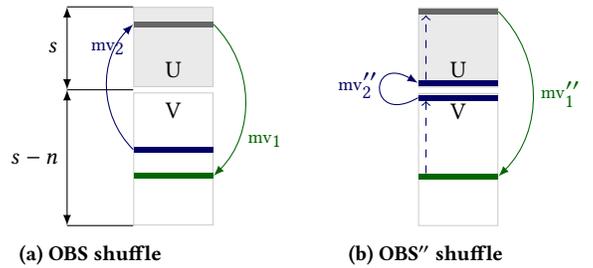
\begin{figure}[t!]
	\begin{center}
		\begin{subfigure}[b]{.30\linewidth}
		\begin{tikzpicture}[every node/.style={node distance=0pt,rectangle,draw,minimum height=32pt,minimum width=30pt,inner sep=5pt,align=center}, scale=1, transform shape]
		
		\node [minimum height=30pt, fill=black!40, opacity=0.2, minimum width=30pt](s){};
		\node [below=2pt of s, opacity=0.2, , minimum height=50pt,minimum width=30pt] (ns){};
		\node [below=-23pt of s, draw=none, minimum width=30pt](sc) {U};	
		\node [above=-23pt of ns, draw=none, minimum width=30pt](nsc) {V};
			
		\node [draw=none,below=-25pt of s,  fill=black!60, minimum height=2pt,inner sep=1pt,minimum width=30pt] (c2a){};
		
		\node [draw=none,below=-30pt of ns, fill=black!60!blue, minimum height=2pt,inner sep=1pt,minimum width=30pt] (c2b) {};
		
		\node [draw=none,below=-20pt of ns, fill=black!60!green, minimum height=2pt,inner sep=1pt,minimum width=30pt] (c1b) {};	
		
		\path[every node/.style={draw=none}] (c2a.east) edge[bend left=45,-latex,black!60!green] node [below right = 10pt and -2pt] {$\text{\footnotesize mv}_1$} (c1b.east);
		\path[every node/.style={draw=none}] (c2b.west) edge[bend left=45,-latex,black!60!blue] node [below left = -20pt and -10pt] {$\text{\footnotesize mv}_2$} (c2a.west);

		\coordinate [above left=0pt and 25pt of s](height4);
		\coordinate [below left=1pt and 25pt of s](height5);
		\coordinate [below left=1pt and 0pt of s](height51);
		\coordinate [below left=0pt and 25pt of ns](height6);
		\draw (s.north west) -- (height4);
		\draw (height51) -- (height5);
		\draw (ns.south west) -- (height6);
		\draw [latex-latex](height4) -- (height5);
		\draw [latex-latex](height5) -- (height6);
				
		\node [draw=none, left=15pt of s](eq1){$s$};
		\node [draw=none, left=22pt of ns](eq1){$s-n$};
		
		
		\end{tikzpicture}
		\caption{OBS shuffle}
		\label{fig:shuffle1}
		\end{subfigure} \hspace{50pt}
		\begin{subfigure}[b]{.30\linewidth}
			\begin{tikzpicture}[every node/.style={node distance=0pt,rectangle,draw,minimum height=32pt,minimum width=30pt,inner sep=5pt,align=center}, scale=1, transform shape]
			
				\node [minimum height=30pt, fill=black!40, opacity=0.2,minimum width=30pt](s){};
				\node [below=2pt of s, opacity=0.2,minimum height=50pt,minimum width=30pt] (ns){};
				\node [below=-23pt of s, draw=none, minimum width=30pt](sc) {U};	
				\node [above=-23pt of ns, draw=none, minimum width=30pt](nsc) {V};
				
				\node [draw=none, below=-30pt of s, fill=black!60, minimum height=2pt,inner sep=1pt,minimum width=30pt] (c1a){};
				
				\node [draw=none,below=-20pt of ns, fill=black!60!green, minimum height=2pt,inner sep=1pt,minimum width=30pt] (c1b) {};
						
				\path[every node/.style={draw=none}] (c1a.east) edge[bend left=45,-latex,black!60!green] node [right] {$\text{\footnotesize mv}''_1$} (c1b.east);
				
				\node [draw=none,above=-3pt of ns, fill=black!60!blue, minimum height=2pt,inner sep=1pt,minimum width=30pt] (c3b) {};
				\node [draw=none,below=-3pt of s, fill=black!60!blue, minimum height=2pt,inner sep=1pt,minimum width=30pt] (c3a) {};
						
				\path[every node/.style={draw=none}] (c3b.west) edge [bend left=120, looseness=10,-latex,black!60!blue] node [below left = -10pt and -2pt] {$\text{\footnotesize mv}''_2$} (c3a.west);
				
				\coordinate [above left=-3pt and -3pt of ns](c1l1);
				\coordinate [above left=0pt and -3pt of c1b](c1l2);
				\draw [dashed, black!60!blue,->] (c1l2) -- (c1l1);		
				\coordinate [above left=-3pt and -3pt of s](c2l1);
				\coordinate [above left=0pt and -3pt of c3a](c2l2);
				\draw [dashed, black!60!blue,->] (c2l2) -- (c2l1);
				
				\end{tikzpicture}
				\caption{OBS$''$ shuffle}
				\label{fig:shuffle2}
		\end{subfigure}
		\caption{\small{The OBS shuffle (a) consists of two moves: $mv_1$, which moves a random card from partition U into a random location of partition V; and $mv_2$, which moves a random card from partition V into the vacated place of partition U. In the OBS$''$ (b) shuffle $mv''_1$ moves the top card in U into a random location of V; and $mv''_2$ which moves the top card in V into the bottom place of U.}}
		\label{fig:shuffle}
	\end{center}
	\vspace{-20pt}
\end{figure}

	This situation is similar to shuffling a deck of $n$ cards, using a "top-to-random" shuffle, a well studied method \cite{aldous1986shuffling,diaconis1992analysis,aldous2002reversible}. In a top-to-random shuffle one repeatedly takes the top card of the deck and inserts it into a random position. In the following we will argue that our shuffle, OBS (Fig \ref{fig:shuffle1}), is no worse than a "top-to-random" shuffle.
	
	In OBS the deck is partitioned into $U$ and $V$. It takes a random card from $U$ and inserts it randomly in $V$, and takes a random card of $V$ and inserts it into the place vacated in $U$. Intuitively, this is no worse than the shuffle OBS$'$ which simply takes a random card in $U$ and inserts it in $V$, and takes the top card in $V$ and inserts it at the bottom of $U$. Now consider the shuffle OBS$''$ (Fig \ref{fig:shuffle2}) which is like OBS$'$ but instead of taking a random card in $U$, we take the top card in $U$. Again, OBS$'$ is no worse than OBS$''$, and therefore OBS is no worse that OBS$''$.
	
	To calculate the expected number of rounds of OBS$''$ needed, we proceed similarly to \cite{aldous1986shuffling}. Consider the bottom card of the deck. One has to wait on average $n-s$ rounds before a card is inserted below it. Then one has to wait $(n-s)/2$ more rounds before a card is again inserted below the original bottom card. Continuing in this way, the number of rounds needed for the original bottom to get to the U/V threshold is 
	\[ (n-s) + \frac{n-s}2 + \frac{n-s}3 +\dots + \frac{n-s}{n-s} \leq 1 + (n-s)\log
	(n-s).\]
	A further $s$ rounds are needed for the original bottom to progress from the U/V threshold to the top of the deck. Thus, the total number of rounds is $(n-s)\log(n-s)+s$. Since OBS is no worse than OBS$''$, we have 
	\begin{equation}
	r_1 \leq 1+ (n-s)\log(n-s)+s.
	\label{eq:r1adv}
	\end{equation}
	
	Calculation of $r_2$ is similar. The first of the $N$ positions is overwritten in one round.
	For the second position to be overwritten, we have to wait $N/(N-1)$ rounds; this is a bit longer,
	because we may accidentally overwrite the first one again. The $i$th position requires us to wait $N/(N-i)$ rounds. Thus, all the positions are overwritten after an expected
	\begin{equation}
	r_2=1+\frac{N}{N-1}+\dots+N \leq 1 + N\log N
	\label{eq:r2exp}
	\end{equation}
	rounds.
	Therefore, $r_1+r_2\leq 2+ (n-s)\log(n-s)+s +  N\log N$.
\end{proof}

\section{Analysis of \schemeacr{} content and access-pattern privacy}
\label{app:privacy}

Our aim now is to analyze privacy offered by \schemeacr{} against a multiple query adversary. To this end, we need to make use of two lemmas, starting with the following lemma.

\begin{lem}
\label{lemma}
Let $S$ be the $(n, N)$ \schemeacr-like store with $N=\{p_0,\ldots,p_{N-1}\}$ positions. Suppose that our $\ACCESSOC$ is run $r$ times. 
Let $no(p_i)$ be the event that the position $p_i$ was not overwritten during the $r$ rounds. The probability that, after $r$ runs of the $\ACCESSOC$ method, at least one of the $N$ positions has not been overwritten (i.e. at least one position survived being overwritten) is:
\begin{equation}
\begin{split}
	 p_{N,r}&=\prob{no(p_0)\vee\ \ldots\vee no(p_{N-1})}=\\
	 &=\sum_{i=1}^{N}{(-1)}^{i+1}\cdot\binom{N}{i}\cdot\left(\frac{N-i}{N}\right)^{r}
\end{split}
\label{eq:pnr}
\end{equation}
\end{lem}

In order to increase its advantage in the privacy game, the adversary uses as leverage its knowledge of what blocks are in which positions in the store, and what blocks are in the OC buffer. As OC runs, the adversary might try to track the movement of blocks. We
can model the adversary's knowledge of what is inside the OC buffer as
a probability distribution. 

Suppose the buffer size is $s$.  The OC buffer contains at most one
occurrence of a given block id. Our goal in the following is to describe the evolution of the adversary's knowledge as a probability distribution. We start by recalling that, at each step, OC selects a block to be added to the buffer. 
If the selected block id is already in the buffer, then the set of block ids in the buffer is not changed. But if the selected block is not currently in the buffer, then one of the buffer's slots is chosen (with uniform probability) and the block id in that slot is evicted, in order to add
the selected one.

\noindent Therefore, the buffer contents will be $X$ after such a step if:
\begin{enumerate}
\item The buffer contents were already $X$, and an element in $X$ was
  selected; or
\item For some $x\in X$ and $e\not\in X$, the buffer contents was
  $(X-\{x\})\cup\{e\}$, and $x$ got selected for adding to the buffer, and $e$ was chosen for
eviction from the buffer and gets put back in the store where $x$ had been.
\end{enumerate}

We view this as a system of states and transitions.
A \emph{state} is the pair $(X,\sel)$, where $X$ is a set of block
ids, and $\sel$ is a function from block ids to reals that abstracts
the store. The value $\sel(x)$ is the
probability that a position in the store chosen uniformly at random
contains block id $x$. 
Since $x$ must occur at
least once in the store, and since every other block must also occur
at least once, $x$ appears at least once and at most $N-n+1$
times. Therefore, 
for any $x$, the value $\sel(x)$ is in the set
$\{ \frac1N, \frac2N, \dots, \frac{N-n+1}N \}$.

There are two kinds of transition, corresponding to the two 
points above:
\begin{enumerate}
\item For each $x\in X$: $(X,\sel)$ can transition to $(X,\sel)$, with probability
  $\sel(x)$.
\item For each $x\in X$ and $e\not\in X$:\\
  $((X-\{x\})\cup\{e\}),\sel\Big[{x\mapsto\sel(x)+\frac1N\atop e\mapsto\sel(e)-\frac1N}\Big])$
  can transition to $(X,\sel)$, with probability
  $\sel(x)\times\frac1s$. (Here, the notation $\sel[x\mapsto expr]$
  means the function $\sel$ but with $x$ mapping to $expr$. The factor
  $\frac1s$ in the transition probability is the
  probability that $e$ is evicted.)
\end{enumerate}
This system is easily seen to be a Markov chain with eigenvalues $1=|\lambda_1|>|\lambda_2|\geq\dots\geq|\lambda_n|$. Moreover, this Markov chain describes a random walk in a strongly connected Eulerian digraph in which nodes represent the possible states $(X,\sel)$ and the edges represent the transitions between them. 

\smallskip

Let us write
$\inb_j(X)$ for the probability that, after the $j$th run, the set 
of block ids present in the buffer is $X$. We can state the following lemma.

\begin{lem}\label{lem:out}
Let $\out_j(B)$ be the probability that the block coming out of an
OC with buffer size $s$ at the $j$-th step is the block $B$.
For all $j,j'>r$ and initial distributions $\inb_0$, the probability $|\out_j(B)-\out_{j'}(B)|$ 
is at most
\[\binom{n-1}{s-1}\cdot C \cdot |\lambda_2|^r\]
for some constant $C$ and $\lambda_2<  1$, and is therefore negligible in $r$.
\end{lem}

\begin{proof}
Let $P$ be the transition matrix for this Markov chain. Then
$P$ is stochastic, irreducible and aperiodic.
Let $\pi_j$ be the probability distribution after $j$ steps. Then
$\pi_j=\pi_0\cdot P^j$.
Let $\lambda_1$ and $\lambda_2$ be the largest and second-largest
eigenvalues of $P$. 
Overwhelmingly likely, the enginvalues are all distinct.
Since $P$ is stochastic, $\lambda_1=1$ and $|\lambda_2|<1$.
Applying the
Perron-Frobenius Theorem similarly to Backaker \cite[Theorem 3 and
Example 1]{backaaker2012google}, we
have $P^r=P^\infty+O(|\lambda_2|^r)$, and therefore 
$\pi_j-\pi_{j'}=\pi_0\cdot C'|\lambda_2|^r$ for $j,j'>r$ for some
constant $C'$.

Let $\pi(X,\sel)$ be the probability that the state is $(X,\sel)$ in
the probability distribution $\pi$.
Then $\inb_j(X)=\sum_{\sel}\pi_j(X,\sel)$.
Let $m_{j,j'}=\max_X(|\inb_j(X)-\inb_{j'}(X)|)$, where 
$X$ ranges over sets of block ids of size  $s$. 
Then, for $j,j'>r$,
\[\begin{array}{rcl}
m_{j,j'}&=&\max_X(\sum_{\sel}|\pi_j(X,\sel)-\pi_{j'}(X,\sel)|)\\
&=& \max_X(\sum_{\sel}(\pi_0(X,\sel)\cdot C'|\lambda_2|^r))\\
&=&  C'|\lambda_2|^r \cdot \max_X(\sum_{\sel}\pi_0(X,\sel)).
\end{array}\]
  
Notice that $\out_j(B)=\sum_{X\mid B\in X}\inb_j(X)$.
Therefore, 
\[\begin{array}{rcl}
|\out_j(B)-\out_{j'}(B)|
&=& \displaystyle\sum_{X\mid B\in X}|\inb_j(X)-\inb_{j'}(X)|\\
 &\leq& \displaystyle\binom{n-1}{s-1} \cdot m_{j,j'}\\
 &\leq&  \displaystyle\binom{n-1}{s-1} \cdot C \cdot |\lambda_2|^{r}\\
  \end{array}
  \]
where $C=C'\cdot \max_X(\sum_{\sel}\pi_0(X,\sel))$.
\end{proof}

\begin{theorem} \schemeacr{} has content and access-pattern privacy, i.e. the advantage of any multiple query adversary against the privacy of \schemeacr{} is negligible in the security parameter $\lambda$ and the number of OC rounds $r$. 
\end{theorem}

\begin{proof} 
	We bound the distinguishing advantage of any \schemeacr{} adversary through the following sequence of ``game hops''.
	
	\smallskip \emph{Game 1.} This is the legitimate ${\Exp}_{m,\adv}^{0}(\schemesymb)$ experiment. In particular, in the $j$-th challenge query $\cdv$ returns $\ACCESSRW(B_{0,j}, op_{0,j},\\ d_{0,j}, S, k)$, for $j=1,\ldots,q$.
	
	We define the following series of experiments for $j=1,\ldots,q$. 
	
	\smallskip \emph{Game ${2,j}$.} Let $j\in\{1,\ldots,q\}$ be fixed.  In this expe\-ri\-ment, the queries $i=1,\ldots,j$ are answered with $\ACCESSRW(B_{0,i}, op_{0,i}, d_{1,i}, S, k)$, namely
	by performing operation $op_{0,i}$ on block $B_{0,i}$ with data $d_{1,i}$ (instead of data $d_{0,i}$). The remaining queries until the $q$-th query are responded normally, i.e. with $\ACCESSRW(B_{0,i}, op_{0,i}, d_{0,i}, S, k)$ for $i=j+1,\ldots,q$.
	
	Let us call this modified experiment ${\Exp}_{m,\adv}^{2,j}(\schemesymb)$. 
	One can easily see that  $$\displaystyle\Big\lvert\prob{{\Exp}_{m,\adv}^{2,j-1}(\schemesymb)=1}- \prob{{\Exp}_{m,\adv}^{2,j}(\schemesymb)=1}\Big\rvert$$
	for $j=1,\ldots,q$ is upper-bounded by the distinguishing advantage against the semantic
	security of the encryption scheme, which is negligible \cite{indcpa}. Trivially ${\Exp}_{m,\adv}^{2,0}(\schemesymb)={\Exp}_{m,\adv}^{0}(\schemesymb)$.
	
	Next, for $j=1,\ldots,q$ we define a new sequence of experiments.
	
	\smallskip \emph{Game ${3,j}$.} Fix $j$. In this experiment, the queries $i=1,\ldots,j$ are answered with $\ACCESSRW(B_{1,i}, op_{0,i}, d_{1,i}, S, k)$,
	i.e. by running operation $op_{0,i}$ on block $B_{1,i}$ with data $d_{1,i}$ (instead of running it on block $B_{0,i}$). The remaining queries until the $q$-th query remain unchanged, i.e. they return $\ACCESSRW(B_{0,i}, op_{0,i},\allowbreak d_{1,i}, S, k)$ for $i=j+1,\ldots,q$. Let us call resulting experiment ${\Exp}_{m,\adv}^{3,j}(\schemesymb)$.
	
	We proceed to upper bound the probability
	that an adversary has in distinguishing access to two different
	blocks $B_{2,j}$ and $B_{1,j}$ in experiment ${\Exp}_{m,\adv}^{3,j}(\schemesymb)$.
	
	Consider the $r$ runs of OC. Let $r_1=r/2$; we will distinguish
	between the first set of $r_1$ runs, and the remaining $r-r_1$ runs.
	Now we distinguish between
	two situations that arise after the runs:
	\setlist[description]{font=\normalfont\itshape}
	\begin{description}
		\item[State A] The adversary has observed that all the positions in
		the store got overwritten by a block id coming out of OC during the
		second set of $r-r_1$ runs.
		\item[State B] The adversary has observed that not all the positions
		got overwritten; that is, at least one position survived being
		overwritten during the second set of $r-r_1$ runs.
	\end{description}
	When a store position is overwritten at step $j$ by a block coming from OC, the
	adversary's probability distribution of what block is in that store
	position is $\out_j(\cdot)$.
	If State A is observed, then the adversary's probability of
	distinguishing the block read in ACCESSRW is at most
	\[\max_{B,i,i'}|\prob{S_i\mbox{ contains }B}-\prob{S_{i'}\mbox{ contains }B}|\]
	which is at most
	$\max_{r_1<j,j'<r\atop B}|\out_j(B)-\out_{j'}(B)|$.
	By Lemma~\ref{lem:out}, this value is at most $\binom{n-1}{s-1}\cdot
	S^r$.
	
	The probability that state $B$ is observed is 
	$p_{N,r_1}$ (this notation is defined in Lemma \ref{lemma}).
	
	The probability of distinguishing is
	the probability of arriving in state A times the probability of
	distinguishing in state A, plus 
	the probability of arriving in state B times the probability of
	distinguishing in state B.
	This is at most:
	
	\begin{equation}
	(1-p_{N,r})\times \binom{n-1}{s-1}\cdot
	S^r+ p_{N,r} \times 1
	\label{eq:adv}
	\end{equation}
	
	\noindent Since both $S^r$ and $p_{N,r}$ are negligible
	in $r$, the probability and therefore the advantage of the adversary is negligible in $r$.
	
	Let us define the last series of game hops, again for $j=1,\ldots,q$.
	
	\smallskip \emph{Game ${4,j}$.}
	In this experiment, the queries $i=1,\ldots,j$ are answered with $\ACCESSRW(B_{1,i}, op_{1,i}, d_{1,i}, S, k)$, i.e.
	by running operation $op_{1,i}$ on block $B_{1,i}$ with data $d_{1,i}$ (instead of running operation $op_{0,i}$). The remaining queries until the $q$-th query remain unchanged, i.e. they return $\ACCESSRW(B_{1,i}, op_{0,i}, d_{1,i}, S, k)$ for $i=j+1,\ldots,q$. Let us call resulting experiment ${\Exp}_{m,\adv}^{4,j}(\schemesymb)$.
	Actually, it holds that $$\prob{{\Exp}_{m,\adv}^{4,j-1}(\schemesymb)=1}=\prob{{\Exp}_{m,\adv}^{4,j}(\schemesymb)=1}$$ by construction. 
	This is because $\ACCESSRW$ performs the same (read position followed by write position) both for read and write operations. It is easy to see that ${\Exp}_{m,\adv}^{4,q}(\schemesymb)={\Exp}_{m,\adv}^{1}(\schemesymb)$.
	
	Finally, by adding the probabilities obtained in each game hop, the statement
	of the theorem follows.
\end{proof}

\section{Invariants}
\label{app:invariants}

We  prove the
following invariant: \emph{for every client and every block, the client has
	a valid position for the block in its map} (INV-2 below).
In the following we will use $block[p]$ notation to denote the block contained at position $p$ in the server store, and $map_\mathcal{C}[B]$ to denote the entry for block id $B$ in the map of client $\mathcal{C}$. 

We start with a simpler invariant, which is a useful lemma.

\paragraph{INV-1.} The variable $block[p].cns$ on the server is at most equal to the number of clients that know that that block is at position $p$. More precisely: for all positions $p$,
\begin{multline}
block[p].cns =\\
\Big|\{\mathcal{C} \in Clients\mid p \in map_\mathcal{C}[block[p].bid].psns\\
{}\land map_\mathcal{C}[block[p].bid].ts = block[p].ts\Big|
\end{multline}

The code maintains this invariant by linking any change to the local map with an operation on the \textit{block[p].cns} value. These changes happen when the contents of the block is updated or deleted (by marking the block as \textit{free}). The SyncPositons operation (see Algorithm \ref{algo:csw-sync}) updates a client's local map to reflect changes performed by other clients and frees old data. The \textit{WriteBlock}
and \textit{DuplicateBlock} set the \textit{block[p].cns} value to 1 in order to trigger map changes in other clients.

As a corollary, we have the following:

\paragraph{INV-1$'$.} 
When $block[p].cns$ has the maxim value ($|Clients|$), then every client's local map contains the latest information about position $p$. More precisely, for each $p$:
\begin{multline}
block[p].cns=|Clients|  \Rightarrow\\
\forall \mathcal{C} \in clients,~p \in map_\mathcal{C}[block[p].bid].psns~\land\\
\land~map_\mathcal{C}[block[p].bid].ts = block[p].ts
\end{multline}

\paragraph{INV-2.} For each block, a valid position is always known to all \schemeacr{} clients.

More precisely, for any block id $B$, there is a position $p$
such that 
\[block[p].bid=B\quad\land\quad p\in \bigcap_{\mathcal{C}\in
	Clients}map_\mathcal{C}[B].psns \]

Before we prove this invariant, we provide some intuition. To maintain INV-2, we use the $map[B].\vf$ set stored in the client's map. This set tracks which are the positions $p$ of a block that a client has observed to have a maximum value for $block[p].cns$. In order to prevent data loss, we require that (1) at any time each client can only reassign a single position, and (2) that at least one position still remains if all the clients decide to reassign one position.

The second part of the requirement (2) is easily achieved by checking that the size of $map[B].\vf$ is bigger than the number of clients. We address the first requirement (1) by requiring each client mark the $map[B].\vf$ set as empty whenever they reassign a position from it. This will prevent the client to reassign any consolidated positions until it re-learns their location.

\begin{proof}
We prove it for the case that there are two clients, say $\mathcal{C}$ and $\mathcal{D}$; as will be seen, the proof generalises intuitively to more clients. Suppose INV-2 is not an invariant; then there is a transition from a state $st_3$ in which INV-2 holds for a block id $B$, to a state $st_4$ in which it does not hold for $B$. Suppose client $\mathcal{C}$ reallocates the crucial position $p$ in $st_3$ which is lost in $st_4$. Since $st_3$ satisfies INV-2, in $st_3$ $p\in map_\mathcal{C}[B].psns$ and $p\in map_\mathcal{D}[B].psns$, and since $st_4$ does not satisfy INV-2, in $st_4$ $map_\mathcal{C}[B].psns \cap map_\mathcal{D}[B].psns=\emptyset$. Using Algorithms \ref{algo:csw-sync} and \ref{algo:csw-duplicate}, we see that the transition for $\mathcal{C}$ from $st_3$ to $st_4$ required $|map_\mathcal{C}[B].\vf|>2$ in $st_3$, so suppose that $map_\mathcal{C}[B].\vf\supseteq\{p,q,r\}$ in $st_3$. Since $\vf\subseteq psns$ (Algorithms \ref{algo:csw-sync} and \ref{algo:csw-duplicate}), we have $\{p,q,r\}\subseteq map_\mathcal{C}[B].psns$ in $st_3$.

Let $st_0$ be the state immediately after the previous reallocation by $\mathcal{C}$. Then $map_\mathcal{C}[B].\vf=\emptyset$ in $st_0$. Since $q,r\in map_\mathcal{C}[B].\vf$ in $st_3$, there was a state between $st_0$ and $st_3$ in which $q.cns=2$, and one in which $r.cns=2$. Let $st_1$ be the most recent of those states. Either $q$ or $r$ was reallocated between $st_1$ and $st_3$, and by definition of $st_0$, that reallocation was done by $\mathcal{D}$. Consider the
most recent reallocation by $\mathcal{D}$ done before $st_3$, say in a state $st_2$. Now we have states $st_0,st_1,st_2,st_3,st_4$ in temporal order. Based on Agorithms \ref{algo:csw-sync} and \ref{algo:csw-duplicate}, $|map_\mathcal{C}[B].\vf|>2$ in $st_2$, so say $map_\mathcal{D}[B].\vf\supseteq\{t_1,t_2,t_3\}$. Then, in $st_2$ we have: $block[t_1].cns=block[t_2].cns=block[t_3].cns=2$, and by INV-1$'$, $\{t_1, t_2, t_3\}\subseteq map_\mathcal{C}[B].psns \cap map_\mathcal{D}[B].psns$. In $st_3$, $\mathcal{D}$'s transition has removed one element from $map_\mathcal{D}[B].psns$, and in $st_4$, $\mathcal{C}$'s transition has removed one element from $map_\mathcal{C}[B].psns$. Therefore, in $st_4$, $map_\mathcal{C}[B].psns \cap map_\mathcal{D}[B].psns$ is non-empty, contradicting our hypothesis.
\hfill
\end{proof}

\section{Algorithms}
\label{app:algorithms}

\removelatexerror
\begin{algorithm2e}[H]
\footnotesize
\SetKwProg{function}{function}{}{}

\KwIn{block id, block, data}
\KwOut{block, status}

\SetKwProg{fn}{function}{}{}

\BlankLine
\function{PrepareWrite (bid, block, data)}{
	\If {$(block.bid = bid$ \textbf{or} $block.bid = free)$}{
		$block.bid \gets bid$\;
		$block.data \gets data$\;
		$block.cns \gets 1$\;
		$block.ts \gets current\_time$\;
		$status \gets block.data$\;
	}
	\Else {
		$status \gets null$\;
	}
	\Return $(block,status)$\;
}

\caption{\small Write a block to the store.}
\label{algo:csw-write}
\end{algorithm2e}

\removelatexerror
\begin{algorithm2e}[H]
\footnotesize
\SetKwProg{function}{function}{}{}

\KwIn{source block, destination block, destination position}
\KwOut{destination block, client map}

\SetKwProg{fn}{function}{}{}

\BlankLine
\function{DuplicateBlock (sblock, dblock, p, $map_c$)}{
	\If 
	{$((dblock.cns=\vert clients\vert$ \textbf{and} $\vert map_c[dblock.bid].vf\vert>\vert clients\vert)$ \textbf{or}\\
	\hspace{3pt} $(dblock.cns=1$ \textbf{and} $\vert map_c[dblock.bid].psns\vert>1))$\\
	}{
		$dblock.bid \gets sblock.bid$\;
		$dblock.data \gets sblock.data$\;
		$dblock.ts \gets sblock.ts$\;
		$dblock.cns \gets 1$\;
		$\mathit{clear}~map_c[dblock.bid].vf$\;
		$\mathit{move}~p~\mathit{from}~map_c[dblock.bid].psns~\mathit{to}~$ $map_c[sblock.bid].psns$\;
	}	
	\Return $(dblock,map_c)$\;
}

\caption{\small Duplicate a store block to a new position.}
\label{algo:csw-duplicate}
\end{algorithm2e}

\removelatexerror
\begin{algorithm2e}[H]
\footnotesize
\SetKwProg{function}{function}{}{}

\KwIn{block, position, buffer, obfuscation client map}
\KwOut{block, buffer}

\SetKwProg{fn}{function}{}{}

\BlankLine
\function{UpdateBuffer($blk$, $p$, $\mathit{buffer}$, $map_{oc}$)}{

	\If {($\mathit{buffer}$ \textbf{not} $full$)}{
		add $blk$ to $\mathit{buffer}$\;
	}

	\If {($\mathit{buffer}$ \textbf{is} $full$)}{
		$buf\_blk \xleftarrow{R} \mathit{buffer}$\;
		$(blk, map_{oc}) \gets DuplicateBlock (buf\_blk, blk, p, map_{oc})$\;
		\If {$(buf\_blk = blk)$}{
			remove $buf\_blk$ from $\mathit{buffer}$\;
		}
	}
	\Return (blk,buffer)\;
}

\caption{\small Update the local buffer data structure.}
\label{algo:csw-oc-update}
\end{algorithm2e}

\end{document}